\documentclass[11pt]{article}
\usepackage[utf8]{inputenc}
\usepackage{amsmath}
\usepackage{amsfonts}
\usepackage{amssymb}
\usepackage{amsthm}
\usepackage{tikz}

\theoremstyle{plain}
\newtheorem{theorem}{Theorem}
\newtheorem{definition}[theorem]{Definition}
\newtheorem{corollary}[theorem]{Corollary}
\newtheorem{lemma}[theorem]{Lemma}
\newtheorem{proposition}[theorem]{Proposition}
\theoremstyle{remark}

\newtheorem{property}{Property}

\title{The graphs with the max-Mader-flow-min-multiway-cut property}
\author{
Vincent Jost\footnote{CNRS, LIX, École Polytechnique, Palaiseau, France, \texttt{vincent.jost@lix.polytechnique.fr}.} \and 
Guyslain Naves\footnote{Department of Mathematics and Statistics, McGill University, Montréal, Canada, \texttt{naves@math.mcgill.ca}.}
}
\date{}

\newcommand{\point}[1]{\node[circle,inner sep = 0pt,minimum size =0pt]}

\newcommand{\Vertex}[1]{\node[circle,inner sep = 0pt,minimum size =5pt,fill = #1]}
\newcommand{\Term}[1]{\node[rectangle,inner sep = 0pt,minimum size =5pt,fill = #1]}

\definecolor{dgreen}{rgb}{0,0.6,0}
\definecolor{dred}{rgb}{0.6,0,0}
\definecolor{dblue}{rgb}{0,0,0.6}
\definecolor{lgray}{rgb}{0.8,0.8,0.8}
\definecolor{lblue}{rgb}{0.6,0.6,1}
\definecolor{lgreen}{rgb}{0.6,1,0.6}
\definecolor{lred}{rgb}{1,0.6,0.6}

\DeclareMathOperator{\repr}{r}

\newcommand{\signed}{\overline{\mathcal{A}_3}}

\begin{document}

\maketitle

% abstract
\begin{abstract}
We are given a graph $G$, an independant set $\mathcal{S} \subset V(G)$ of \emph{terminals}, and a function $w:V(G) \to \mathbb{N}$. We want to know if the maximum $w$-packing of vertex-disjoint paths with extremities in $\mathcal{S}$ is equal to the minimum weight of a vertex-cut separating $\mathcal{S}$. We call \emph{Mader-Mengerian} the graphs with this property for each independant set $\mathcal{S}$ and each weight function $w$. We give a characterization of these graphs in term of forbidden minors, as well as a recognition algorithm and a simple algorithm to find maximum packing of paths and minimum multicuts in those graphs.  
\end{abstract}

\section{Introduction}

% blabla

Given a graph $G=(V,E)$, a set $\mathcal{S} \subset V$ with $|\mathcal{S}| \geq$ $2$ and inducing a stable set is called a set of \emph{terminals}. An \emph{$\mathcal{S}$-path} is a path having distinct ends in $\mathcal{S}$, but inner nodes in $V \setminus \mathcal{S}$. A set ${\mathcal{P}}$ of $\mathcal{S}$-paths, is a \emph{packing of vertex-disjoint $\mathcal{S}$-paths} (since there is no risk of confusion, we will use the shorter term \emph{packing of $\mathcal{S}$-paths} within this paper), if two paths in ${\mathcal{P}}$ do not have a vertex in common in $V \setminus \mathcal{S}$. We are looking for a maximum number $\nu(G,\mathcal{S})$ of $\mathcal{S}$-paths in a packing.

An \emph{${\mathcal{S}}$-cut} is a set of vertices in $V \setminus \mathcal{S}$ that disconnect all the pairs of vertices in $\mathcal{S}$ (that is a blocker of the $\mathcal{S}$-paths). We are looking for an ${\mathcal{S}}$-cut with a minimum number $\kappa(G,\mathcal{S})$ of vertices.

The following inequality holds for any graph $G$ and any ${\mathcal{S}}\subseteq V(G)$: $\nu(G,\mathcal{S})\leq \kappa(G,\mathcal{S})$, as any $\mathcal{S}$-path intersects any $\mathcal{S}$-cut. Note that if $|\mathcal{S}| =$ $2$ the equality always holds, being Menger's vertex-disjoint undirected $(s,t)$-paths theorem. This paper deals with graphs for which $\nu(G,{\mathcal{S}})=\kappa(G,{\mathcal{S}})$, for any set $\mathcal{S}$ of terminals. Actually, we try to characterize a stronger property associated with a weighted version of these two optimization problems. Consider the following system with variables $x\in \mathbb{R}^{V\setminus \mathcal{S}}_+$:

\begin{equation}\label{eqn:blocking}
x(P)\geq 1 \mbox{ for every } {\mathcal{S}}\mbox{-path } P
\end{equation}

An integral vector $x$ minimizing $wx$ over~\eqref{eqn:blocking} is necessarily
a ${0,1}$-vector and is the characteristic vector of a minimum
$\mathcal{S}$-cut. Dually, an integral vector $y$ optimum for the dual of
minimizing $wx$ over~(\ref{eqn:blocking}) is necessarily a maximum $w$-packing
of $\mathcal{S}$-paths. Hence, if~\eqref{eqn:blocking} is a TDI system, we
have that the minimum $w$-capacity of an $\mathcal{S}$-vertex-cut is equal to
the maximum $w$-packing of ${\mathcal{S}}$-paths.

\begin{figure}
\begin{center}
\begin{tikzpicture}[x=1cm,y=1cm]
\Vertex{black} (a1) at (90:1) {};
\Vertex{black} (b1) at (210:1) {};
\Vertex{black} (c1) at (330:1) {};
\Term{black} (a2) at (90:2) {};
\Term{black} (b2) at (210:2) {};
\Term{black} (c2) at (330:2) {};
\draw[black] (a2) -- (a1) -- (b1) -- (b2);
\draw[black] (b1) -- (c1) -- (a1);
\draw[black] (c1) -- (c2);
\end{tikzpicture}
\end{center}
\caption{The net.}
\label{fig:net}
\end{figure}

As an example, consider the graph of Figure~\ref{fig:net}, called
\emph{net}. Let $\mathcal{S}$ be the square vertices. A maximum integral
packing of $\mathcal{S}$-paths ($w=1$) contains only one path, while any
$\mathcal{S}$-cut must contain at least two vertices. Precisely, there is a fractional
packing of $\mathcal{S}$-paths of value $\frac{3}{2}$ (by taking each
$\mathcal{S}$-path of length $3$ with value $\frac{1}{2}$), and a fractional
$\mathcal{S}$-cut with the same value (by taking $x(v) = \frac{1}{2}$ for all
$v \notin \mathcal{S}$).

Motivated by the following property, we call \emph{Mader-Mengerian} the
graphs for which the system \eqref{eqn:blocking} is TDI for every set
${\mathcal{S}}$ of terminals.

\begin{property}\label{lemma:perfect}
Given a graph $G$ and a set of terminal $\mathcal{S}$, the following conditions are equivalent:
\begin{enumerate}
\item The system~\eqref{eqn:blocking} is TDI,
\item The polyhedron defined by~\eqref{eqn:blocking} is integral,
\item The optimum value of maximizing $w^Tx$ subject to~\eqref{eqn:blocking} is integral (if finite) for all $w \in V^{ \{0,1,+\infty\} }$.
\end{enumerate}
\end{property}

The proof of this property is postponed to section~\ref{sec:bipartite} where
the stronger Lemma~\ref{lemma:main} is proved. We already know that the long
claw is not Mader-Mengerian.

Our main result (Theorem~\ref{th:bad-graphs}) is a description of the Mader-Mengerian graphs in terms of forbidden minors. However we do not use the usual minor operations (edge deletion and edge contraction), but \emph{ad-hoc} operations on vertices. Our proof implies an algorithm (Lemma~\ref{lemma:main}) to find maximal $w$-packing of paths in Mader-Mengerian graphs and minimum vertex multicuts for a given set of terminals. We also give a characterization of the pairs $(G,\mathcal{S})$ for which the system~\eqref{eqn:blocking} is TDI (Theorem~\ref{th:signed}).

One of our most surprising results is that $G$ is Mader-Mengerian if and only if the system~\eqref{eqn:blocking} is TDI for every independant set $\mathcal{S}$ of cardinality $3$. This implies (with Lemma~\ref{lemma:main}) a polynomial algorithm to recognize Mader-Mengerian graphs.\medskip

% Multiway cuts
Finding a minimum $\mathcal{S}$-cut is an NP-complete problem, even if $|\mathcal{S}|=3$~\cite{papaseym}. In fact,~\cite{papaseym} deals with edge-cuts (that is, sets of edges disconnecting $\mathcal{S}$), but one may observe that $\mathcal{S}$-edge-cut in a graph $G$ correspond to vertex-cut in the line-graph of the graph obtained from $G$ by adding one leaf to each vertex in $\mathcal{S}$.   

% Mader's theorem and its variants.

Finding maximal packing of disjoint paths is a classical problem in graph theory, even if it was mainly studied for edge-disjoint (or arc-disjoint) paths. Menger~\cite{menger} gave the first significant result, stating that when $|\mathcal{S}|=2$, the maximum number of disjoint $S$-paths is equal to the minimum cardinality of an $(s,t)$-cut, both in edge-disjoint and vertex-disjoint cases. This result was further developped by Ford and Fulkerson~\cite{fordfulkerson}, into what became the network flow theory. When there is more than two terminals, the results are however closer to matching theory than to network flows. Gallai~\cite{gallai} first proved a min-max theorem for packing of fully-disjoint $\mathcal{S}$-paths (that is even the ends of the paths must be disjoint), and his result was then strengthened by Mader~\cite{mader} for inner-disjoint paths with ends in different parts of a partition of the terminals. Mader's theorem implies the following:
\begin{theorem}[Mader, 1978]
Let $G$ be a graph and $\mathcal{S}$ an independant set of $G$. Then,
\[
 \nu(G,\mathcal{S}) = 
 \min |U_0| + \sum_{i=1}^k \left\lfloor \frac{b_{U_0}(U_i)}{2} \right\rfloor 
\]
where the minimum ranges over all the partitions $U_0,\ldots,U_k$ of $V \setminus \mathcal{S}$ , such that each $S$-path intersects either $U_0$ or $E(U_i)$ for some $1 \leq i \leq k$. Here, $b_{U_0}(X) := |\{v \in X~:~N(v) \setminus (X \cup U_o) \neq \emptyset\}|$.  
\end{theorem}
In the light of Mader's theorem, we are looking for graphs that admit a much simpler characterization: $\nu(G,\mathcal{S}) = \min |U|$ where the minimum ranges over sets $U$ such that each $S$-path intersects $U$. A practical reason for looking for these graphs is that Mader's theorem relies on matching theory, while our result will only use Menger's theorem, that is flow theory. As a consequence, algorithms for finding an optimal packing of $\mathcal{S}$-paths in Mader-Mengerian graphs are simpler and more efficient than those for general graphs. 

Mader's theorem has been recently extended by Chudnovsky et al.~\cite{chudnovskyetal}, and by Gyula Pap~\cite{pap}.

Let us mention a similar result for edge-disjoint paths, that was proved by Cherkasky~\cite{cherkasky} and Lovász~\cite{lovasz}:
\begin{theorem}[Cherkasky, Lovász, 1977]
For any inner Eulerian graph $G$, then the maximum number of edge-disjoint $\mathcal{S}$-paths is equal to $\frac{1}{2} \sum_{s \in \mathcal{S}} \lambda{s}$, where $\lambda{s}$ is the minimum cardinality of a cut between $s$ and $\mathcal{S} - s$.
\end{theorem}
This has been later extended by Karzanov and Lomonosov~\cite{karzanosov}, who proved the Locking Theorem. %A subset $S$ of the terminals $\mathcal{S}$ is locked by a family $P$ of edge-disjoint $\mathcal{S}$-paths if $P$ contains $\lambda{S}$ $(S,\mathcal{S}-S)$-paths. The Locking Theorem characterizes families $\mathcal{F}$ of $2^{\mathcal{S}}$ for which a maximal packing of $\mathcal{S}$-paths lock every member of $\mathcal{F}$.
These results explains when the maximum packing of edge-disjoint $\mathcal{S}$-paths has a characterization in terms of minimal cuts.

% about finding special cases with simple characterization (flow-based, say) when we only know a not-so-simple characterization (matching-based). Algorithmic consequences. Actually there is a lot to say about that. Even in ``classic CO'', like bipartite matching, shortest paths with non-negative length, colouring subclasses of perfect graphs, cut condition in disjoint paths problems, we can probably find many good examples...

%our results

%plan of the paper.

\section{Vertex minors and skew minors}
% First, is there some litterature about these operations ?
% If yes, some of the following questions may be irrelevant.

% Why vertex minors? Which problems could be candidate for having nice vertex minor characterization (node-capacitated connectivity problems only ?)

% vertex minors and graph classes. Known classes (interval, cocomp, AT-free, proofs in my thesis). families with/without WQO property.

%%% The minors operations.
Given a graph $G=(V,E)$ and $v\in V$, \emph{deleting} $v$ in $G$ means considering the graph $G-v$ induced by $V-v$, that is:
$$G-v:=(V - v, E \setminus \delta_G(v))$$
\emph{Contracting} $v$ means considering the graph $G / v$ obtained by removing $v$ and replacing its neighborhood by a clique:
$$G/v:=(V - v, E \cup \{wx | w, x \in N_G(v)\} \setminus \delta_G(v))$$
For $e=xy \in E$ \emph{contracting} $e$ means considering the graph $G / e$ obtained by identifying
the end-nodes $x$ and $y$ of $e$. 

$$G/e:=(V, E \cup \{xz | z \in N_G(y)\} \cup \{yz | z \in N_G(x)\} \setminus e)$$

A graph obtained from $G$ by any sequence of vertex deletions and vertex contractions is 
a \emph{vertex-minor} of $G$. A graph obtained from $G$ by any sequence of vertex deletions, vertex contractions and edge contractions is a \emph{skew-minor} of $G$.

% proof of commutativity of contraction and deletion
%%%V je comprends pas... pourquoi ne pas annoncer la commutativité
%sans preuve ou avec une preuve du fait lui-meme ? L'enoncé me semble pas très clair...

%G : qu'est-ce qui n'est pas clair ? Je pense que ce lemme permet de donner un
%peu de sens aux opérations, et donc a de l'intérêt en soi.

Vertex-minors can also be described in the following way:

\begin{proposition}
Let $G$ be a graph, and $G'$ be a vertex-minor of $G$. Let $D$ be the vertices deleted and $C$ be the vertices contracted to get $G'$ from $G$. Then, $u, v \in V(G')$ are adjacent in $G'$ if and only if there is a path with extremities $u$ and $v$ in $G$ and whose inner nodes are in $C$.\qed
\end{proposition}

%% \begin{proof}
%% By induction on the number of minor operations. This is certainly true after no operation and it is obvious for a vertex deletion. Suppose $G_0$ is obtained after $n$ operations, and that $u$ and $v$ are not adjacent in $G_0$ and are adjacent in $G_0 / w$. By induction there are two paths in $G$, $P_u$ between $u$ and $w$, and $P_v$ between $v$ and $w$, whose inner nodes were contracted to get $G_0$. The property is then proved by considering the concatenation of $P_u$ and $P_v$.
%% \end{proof}

This immediately implies:

\begin{lemma}\label{lemma:commutativity}
Vertex-deletions and vertex-contractions commute.\qed
\end{lemma}

By definition, for a class of graph, being closed under skew minors implies being closed under vertex minors, which in turn implies being closed under induced subgraphs.
Several important classes of graphs are closed under skew minors. Among them:

\begin{definition}$\quad$\vspace{-10pt}
\begin{itemize}
\item[-] The \emph{interval graphs} are the graphs of intersection of
  intervals of the real line.\vspace{-10pt}
\item[-] The \emph{chordal graphs} are the graphs of intersection of subtree
  of a tree. Equivalently, a graph is chordal if each of its cycles of length
  at least $4$ has a chord.\vspace{-10pt}
\item[-] The \emph{cocomparability graphs} are the graphs whose complement is
  the underlying graph of a partially ordered set.\vspace{-10pt}
\item[-] The \emph{Asteroidal-Triple-free (AT-free) graphs} are the graphs without
  asteroidal triple. A stable set $S$ of cardinality $3$ is an
  \emph{asteroidal triple} of $G$ if there is no $x \in S$ such that $S-x$ is
  contained in a connected component of $G - (x \cup N(x))$.\vspace{-10pt}
\item[-] The $P_k$-free graphs, for $k \in \mathbb{N}$, are the graphs with no
  induced path of length at least $k$.
\end{itemize}
\end{definition}

The following proposition is left as an exercise:

\begin{proposition}\label{lemma:closeness}
Interval graphs, chordal graphs, co-comparability graphs, AT-free graphs,
$P_k$-free graphs are closed under skew minors.\qed
\end{proposition}

The following lemma explains why we are interested in the vertex-minor operations.

\begin{lemma}
Given a graph $G$ and a set of terminal $\mathcal{S}$, if the
system~\eqref{eqn:blocking} is TDI, then it is also TDI for any vertex-minor of $G$.
\end{lemma}

\begin{proof}
Deleting $v \in V \setminus  {\mathcal{S}}$ corresponds to setting $w_v=0$.
Contracting $v \in V \setminus {\mathcal{S}}$ corresponds to setting $w_v=+\infty$.
\end{proof}

\section{Integrality of the blocker of S-paths}\label{sec:bipartite}
%This section is for proving the TDIness of the blocker of S-paths system.

%Definition of the auxiliary graph

%statement and proof 

For a given graph $G$ and a set $\mathcal{S}$ of terminal, we construct an
auxiliary graph $G_{\mathcal{S}}$ as follows. First, note that if a
non-terminal vertex $v$ is adjacent to two terminals $s$ and $t$, we may
assume that the maximum packing for a weight function $w$ contains $w(v)$
times the $2$-length paths $sv,vt$, and the minimal $\mathcal{S}$-cut contains
$v$. Hence, we first delete every non-terminal vertex adjacent to two or more
terminals. 

We may also assume that no $\mathcal{S}$-path of a maximum packing contains
two vertices of $N_G(s)$ for some terminal $s$ (by taking chordless
paths). Therefore if $G - N_G(s)$ contains a component disjoint from
$\mathcal{S}$, we can delete all its vertices.

From now, we will always suppose that:
\begin{itemize}
\item[$(i)$] $G$ has no vertices adjacent to two distinct terminals.
\item[$(ii)$] for each $s \in \mathcal{S}$, every component of $G - N_G(s)$ intersects $\mathcal{S}$.
\item[$(iii)$] $G$ has no edge whose ends are both adjacent to the same terminal. 
\end{itemize}

Then we consider the set $N = N_{G}(\mathcal{S})$ of vertices adjacent to
$\mathcal{S}$. $N$ is the vertex set of $G_{\mathcal{S}}$. We delete the
terminals, and contract the vertices in $V - (N \cup \mathcal{S})$. Then we
remove all the edges whose ends are adjacent to the same terminal in $G$ (the
contraction of a path of a maximum packing would not use these edges) . This
gives $G_{\mathcal{S}}$. By construction, this graph is
$|\mathcal{S}|$-partite, each part being the neighborhood of one terminal.

Note that $a, b \in N$ are adjacent in $G_{\mathcal{S}}$ if $a$ and $b$ are not
adjacent to a common terminal, and there is an $(a,b)$-path in $G$ whose inner
vertices are outside $\mathcal{S} \cup N_G(\mathcal{S})$. %We will denote such a path by $P^{ab}$.

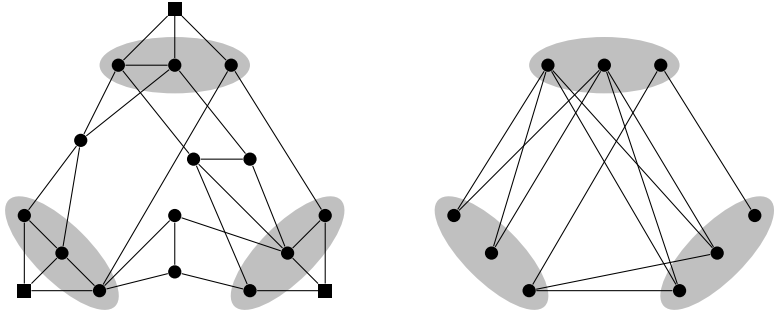
\begin{figure}
\begin{center}
\begin{tabular}{c @{\hspace{1cm}} c}
\begin{tikzpicture}[x=0.25cm,y=0.25cm]
% Fichier d'origine : blocking-auxiliary.pgf
\fill[black,nearly transparent] (28,40) ellipse (4 and 1.5);
\fill[black,nearly transparent,rotate around={135:(22,30)}] (22,30) ellipse (4 and 1.5);
\fill[black,nearly transparent,rotate around={45:(34,30)}] (34,30) ellipse (4 and 1.5);
\Vertex{black} (nx20y32) at (20,32) {};
\Vertex{black} (nx24y28) at (24,28) {};
\Vertex{black} (nx28y32) at (28,32) {};
\Vertex{black} (nx28y29) at (28,29) {};
\Vertex{black} (nx23y36) at (23,36) {};
\Vertex{black} (nx22y30) at (22,30) {};
\Term{black} (nx20y28) at (20,28) {};
\Vertex{black} (nx29y35) at (29,35) {};
\Vertex{black} (nx32y35) at (32,35) {};
\Vertex{black} (nx25y40) at (25,40) {};
\Vertex{black} (nx28y40) at (28,40) {};
\Term{black} (nx36y28) at (36,28) {};
\Vertex{black} (nx36y32) at (36,32) {};
\Vertex{black} (nx34y30) at (34,30) {};
\Vertex{black} (nx32y28) at (32,28) {};
\Term{black} (nx28y43) at (28,43) {};
\Vertex{black} (nx31y40) at (31,40) {};
\draw[black] (nx22y30) -- (nx24y28);\draw[black] (nx34y30) -- (nx36y32);
\draw[black] (nx22y30) -- (nx20y32);\draw[black] (nx29y35) -- (nx32y28);
\draw[black] (nx29y35) -- (nx34y30);\draw[black] (nx32y35) -- (nx34y30);
\draw[black] (nx29y35) -- (nx32y35);\draw[black] (nx25y40) -- (nx29y35);
\draw[black] (nx28y40) -- (nx32y35);\draw[black] (nx31y40) -- (nx36y32);
\draw[black] (nx28y40) -- (nx25y40);\draw[black] (nx28y32) -- (nx34y30);
\draw[black] (nx28y32) -- (nx24y28);\draw[black] (nx28y29) -- (nx32y28);
\draw[black] (nx28y29) -- (nx28y32);\draw[black] (nx24y28) -- (nx28y29);
\draw[black] (nx31y40) -- (nx24y28);\draw[black] (nx28y40) -- (nx23y36);
\draw[black] (nx23y36) -- (nx25y40);\draw[black] (nx23y36) -- (nx22y30);
\draw[black] (nx20y32) -- (nx23y36);\draw[black] (nx31y40) -- (nx28y43);
\draw[black] (nx28y43) -- (nx25y40);\draw[black] (nx28y40) -- (nx28y43);
\draw[black] (nx36y32) -- (nx36y28);\draw[black] (nx36y28) -- (nx34y30);
\draw[black] (nx32y28) -- (nx36y28);\draw[black] (nx20y28) -- (nx24y28);
\draw[black] (nx20y32) -- (nx20y28);\draw[black] (nx20y28) -- (nx22y30);
\end{tikzpicture}
&
\begin{tikzpicture}[x=0.25cm,y=0.25cm]
% Fichier d'origine : blocking-auxiliary-2.pgf
\fill[black,nearly transparent] (11,41) ellipse (4 and 1.5);
\fill[black,nearly transparent,rotate around={135:(5,31)}] (5,31) ellipse (4 and 1.5);
\fill[black,nearly transparent,rotate around={45:(17,31)}] (17,31) ellipse (4 and 1.5);
\Vertex{black} (nx8y41) at (8,41) {};
\Vertex{black} (nx11y41) at (11,41) {};
\Vertex{black} (nx14y41) at (14,41) {};
\Vertex{black} (nx15y29) at (15,29) {};
\Vertex{black} (nx19y33) at (19,33) {};
\Vertex{black} (nx17y31) at (17,31) {};
\Vertex{black} (nx7y29) at (7,29) {};
\Vertex{black} (nx5y31) at (5,31) {};
\Vertex{black} (nx3y33) at (3,33) {};
\draw[black] (nx7y29) -- (nx15y29);\draw[black] (nx17y31) -- (nx7y29);
\draw[black] (nx7y29) -- (nx14y41);\draw[black] (nx5y31) -- (nx8y41);
\draw[black] (nx11y41) -- (nx5y31);\draw[black] (nx3y33) -- (nx11y41);
\draw[black] (nx8y41) -- (nx3y33);\draw[black] (nx15y29) -- (nx8y41);
\draw[black] (nx11y41) -- (nx15y29);\draw[black] (nx17y31) -- (nx11y41);
\draw[black] (nx8y41) -- (nx17y31);\draw[black] (nx14y41) -- (nx19y33);
\end{tikzpicture}
\end{tabular}
\end{center}
\caption{A graph $G$ and the auxiliary graph $G_{\mathcal{S}}$.}
\label{fig:auxiliary}
\end{figure}

\begin{lemma}\label{lemma:main} %voir meme theoreme en fait, non ?
Given a graph $G$ and a set of terminal $\mathcal{S}$, 
the system~\eqref{eqn:blocking} is TDI
if and only if the auxiliary graph $G_{\mathcal{S}}$ is bipartite.
\end{lemma}

\begin{proof}
Assume that $G_{\mathcal{S}}$ is not bipartite. Let $C^\ast$ be an induced odd cycle of $G_{\mathcal{S}}$.
We define a weight vector $w \in V(G)^{ \{0,1,+\infty\} }$ as follows:
\begin{equation}\label{eqn:auxiliary}
w_v:= \left\{\begin{array}{ll}
1       &  \textrm{if } v \in C^* \\
0       &  \textrm{if } v \in V(G_{\mathcal{S}}) \setminus C^* \\
+\infty &  \textrm{otherwise}
\end{array}\right.
\end{equation}

To every edge $uv$ of $C^\ast$, we can associate an $\mathcal{S}$-path of $G$
intersecting $N$ exactly in $u$ and $v$. Then a maximum fractional $w$-packing
of ${\mathcal{S}}$-paths is given by taking $1/2$ for each of these paths and
a minimum fractional $\mathcal{S}$-cut of $G$ is given by $1/2$ on every node
of $C^\ast$, and $1$ on other vertices of $N$. The optimum value of the corresponding pair
of dual linear programs is then $|V(C^\ast)|/2$, hence the polyhedron
defined by~\eqref{eqn:blocking} is not integer.

%Then $\nu(G,w,{\mathcal{S}})=\nu(G_{\mathcal{S}})$ and
%$\kappa(G,w,{\mathcal{S}})=\kappa(G_{{\mathcal{S}}})$

Suppose now that $G_{\mathcal{S}}$ is bipartite, with bipartition $(A,B)$. 

Let $H$ be the graph obtained by deleting $\mathcal{S}$ and add two new
non-adjacent vertices $s_a$ and $s_b$, adjacent to respectively $A$ and
$B$.

Let $P$ be a chordless $(s_a,s_b)$-path in $H$. Let $\{a,b\} := N \cap
V(P)$. We can associate a unique path $\hat{P}$ of $G$ to $P$, by replacing
its extremities by terminals of $G$ (because each vertex of $N$ is adjacent to
a unique terminal). We show that $\hat{P}$ cannot be a cycle. Let $Q =
V(\hat{P}) \setminus (\mathcal{S} \cup N)$. If $Q$ is empty, $\hat{P}$ is
clearly not a cycle because in $H$, the neighborhood of a terminal is a stable
set.

Else $Q$ is contained in a component $C$ of $G \setminus (N \cup
\mathcal{S})$. $C$ is adjacent to $N_G(s)$ and $N_G(t)$ for two distinct
terminals $s$ and $t$ by condition $(ii)$. We can suppose that $a \in
N_G(s)$. $\hat{P}$ is a cycle only if $b \in N_G(s)$. But if this was the
case, then for $c \in N_G(t)$ adjacent to $C$, $a,c,b$ would be a path in $H$,
hence $a$ and $b$ would be in the same part of the bipartition $(A,B)$,
contradiction. $\hat{P}$ is not a cycle, it is an $\mathcal{S}$-path.

By applying the vertex-disjoint version of Menger's theorem to $H$,
$\nu(G,w,\mathcal{S}) = \kappa(G,w,\mathcal{S})$ for any $w\in
\mathbb{Z}^{V\setminus{\mathcal{S}}}$.
\end{proof}

\section{A forbidden minor characterization}\label{sec:minor-charac}
% This section is for characterizing the class by excluded vertex minors, and then (vertex+edge) minors-> rocket-free graphs. 

%many many lemmas

% Les lemmes suivants permettent de réduire l'ensemble des graphes à exclure. Nous montrons que $\mathcal{A}_3$ suffit. Dans les lemmes suivants, $G$ est un graphe de $\mathcal{A}_k$ pour un $k$ impair.

% short description of the situation. To be rewrite when we are more advanced.

In this section, we find a characterization of Mader-Mengerian graphs by excluded vertex-minors. We start from the proof of Lemma~\ref{lemma:main}, where we showed that if a graph is not Mader-Mengerian, its auxiliary graph has an odd cycle. In the auxiliary graph construction, we perform vertex-minor operations plus deletion of edges between two vertices adjacent to the same terminal. It follows that a graph that is not Mader-Mengerian contains a vertex-minor $G$ of the following form.

$G$ is a graph obtained by taking an odd cycle $C$ and the terminals adjacent to $C$. Each vertex of $C$ is adjacent to exactly one terminal, called the \emph{representant} of this vertex. We color the vertices depending on their representants: each representant gets a distinct color, each other vertex has the color of its representant. A color is thus a set of vertices adjacent to some terminal, plus this terminal. Two consecutive vertices of the odd cycle have distinct colors, while the extremities of each chord share the same color. Let $\mathcal{A}_n$ be the class of graphs obtained in this way with $n$ terminals.

One path of lemmas and proofs to obtain the following result is presented
in the Appendix.

\begin{theorem}\label{th:bad-graphs}
Let $G$ be a graph.%%%V, $T$ a stable set of $G$. 
The system~\eqref{eqn:blocking} is TDI for every stable set $\mathcal{S}$ if and only if $G$ does not contain a vertex minor in $\mathcal{A}_3$.
\end{theorem}

\begin{proof}
Direct consequence of Lemmas~\ref{lemma:odd-colors},~\ref{lemma:7colors} and~\ref{lemma:5colors}.
\end{proof}

% Corollary : the graphs for which $\eqref{eqn:blocking}$ is TDI for every $\mathcal{S}$ are the graphs for which it is TDI for every $\mathcal{S}$ of size $3$.

\begin{corollary}
System~\eqref{eqn:blocking} is TDI for every stable set $\mathcal{S}$ of $G$ if and only if it is TDI for every stable set $\mathcal{S}$ of cardinality $3$ of $G$.\qed
\end{corollary} 

This gives a polynomial-time recognition algorithm for the related class of graphs, in combination with Lemma~\ref{lemma:main}: we only have to check for each independant subset of three vertices whether the associated auxiliary graph is bipartite. Another important consequence is that the class of graphs for which system~\eqref{eqn:blocking} is TDI for every stable set is large. Indeed, it contains at least the asteroidal-triple-free graphs:

\begin{corollary}
For every asteroidal-triple-free graph, the system~\eqref{eqn:blocking} is TDI.
\end{corollary}

\begin{proof}
Follows from Theorem~\ref{th:bad-graphs} and Proposition~\ref{lemma:closeness}, as every graph in $\mathcal{A}_3$ contains an asteroidal triple, namely the set of terminals.
\end{proof}

To conclude this section on vertex-minors, we prove that there is an infinite number of minimal graphs to exclude.

\begin{lemma}\label{lemma:infinite}
If $n=3$ and each color class induces a clique, then $G$ is a minimal excluded graph.
\end{lemma}

\begin{proof}
Let $U, V, W \subset V(G)$ be the three colors of $G$, Let $u$, $v$, $w$ be
the three terminals of a minimal excluded minor $G' = G - D / C$. The distance
between two terminals in $G'$ is at least $3$, in particular they cannot be
adjacent. If $x, y \in V(G')$ and $xy \in E(G)$ then $xy \in E(G')$, thus $u$,
$v$ and $w$ have distinct colors in $G$, say $u \in U$, $v \in V$, $w \in W$.

Let $U'$, $V'$ and $W'$ be the color classes of $u$, $v$ and $w$ respectively
in $G'$. Every vertex adjacent to $u$ in $G'$ must be in the same color class
$U'$ as $u$ in $G'$, proving that $U \setminus (D \cup C) \subset U'$. Because
color classes are a partition of the vertex set, we have equality, $U' =U
\setminus (C \cup D)$ and similary for $V'$ and $W'$.

Suppose $C$ is not empty, let $x \in C$. We may assume $x \in U$. If $x$ is
the representant of $U$, then $G' = G - (D + x) / (C - x)$. Else, if $x$ is not
the representant of $U$, $x$ has exactly two neighbors $y$ and $z$ outside
$U$. Because $u$, $v$ and $w$ must be at distance $3$ of each other in $G'$,
$y, z$ must be in $D$. Then we also have that $G' = G - (D + x) / (C -
x)$. Hence $G' = G - (C \cup D)$. But then, as the set of edges between colors
of $G'$ must be a cycle, $G' = G$, proving that $G$ is minimal.
\end{proof}

\section{Minimal skew-minors exclusion}\label{sec:skewminor}

A skew-minor of a graph $G$ is any graph obtained from $G$ via the following operations:
vertex deletion, vertex contraction and edge contraction.

Note that Mader-Mengerian graphs are not closed under edge contraction since
by inflating one of the central vertex of the net we get a Mader-Mengerian
graph. However, we can get a simple sufficient condition for the integrality
of system~\eqref{eqn:blocking} based on skew-minors:

%the super short and factorized proof that we wrote at LIX, or an even better one, when we have edge deletion.

\begin{theorem}\label{coro:rocket}
Any graph $G$ is either Mader-Mengerian or contains a net or a rocket as a skew minor.

%Soit $G$ un graphe ne contenant pas de longue griffe ni de fusée par mineur de sommet et contraction d'arête, et $S$ un ensemble indépendant de $G$, alors le système~\eqref{eqn:S-chemins} est TDI.
%% 
\end{theorem}

\begin{figure}
\begin{center}
\begin{tikzpicture}[x=1cm,y=1cm]
\node[circle,inner sep = 0pt,minimum size =3pt,fill = black] (a) at (90:1.5) {};
\node[circle,inner sep = 0pt,minimum size =3pt,fill = black] (b) at (162:0.7) {};
\node[circle,inner sep = 0pt,minimum size =3pt,fill = black] (c) at (234:1) {};
\node[circle,inner sep = 0pt,minimum size =3pt,fill = black] (d) at (306:1) {};
\node[circle,inner sep = 0pt,minimum size =3pt,fill = black] (e) at (18:0.7) {};
\node[circle,inner sep = 0pt,minimum size =3pt,fill = black] (f) at (230:2) {};
\node[circle,inner sep = 0pt,minimum size =3pt,fill = black] (g) at (310:2) {};
\node[circle,inner sep = 0pt,minimum size =3pt,fill = black] (h) at (90:2.5) {};
\draw[black] (a) to[relative,out=340,in=200] (b); 
\draw[black] (b) -- (d);
\draw[black] (d) -- (c);
\draw[black] (c) -- (e);
\draw[black] (e) to[relative,out=340,in=200] (a);
\draw[black] (a) -- (h);
\draw[black] (b) to[relative,out=340,in=200] (f);
\draw[black] (f) to[relative,out=10,in=170] (c);
\draw[black] (d) to[relative,out=10,in=170] (g);
\draw[black] (g) to[relative,out=340,in=200] (e);
\end{tikzpicture}
\end{center}
\caption{The rocket}
\label{fig:rocket}
\end{figure}
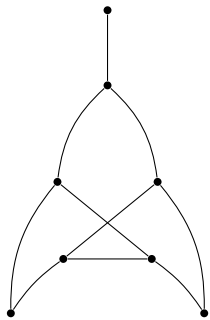

\section{When the set of terminals is fixed}

Our arguments apply when we want to find the pairs $(G,\mathcal{S})$,
$\mathcal{S} \subset V(G)$, for which the system~\eqref{eqn:blocking} is
TDI. Up to now, we have only looked at graphs $G$ for which we have TDIness
for every set of terminals.  To deal with a fixed set of terminals, we define
another notion of vertex-minor, the \emph{signed vertex-minor}, defined on
pairs $(G,\mathcal{S})$. Signed vertex-minor are defined like vertex-minor,
except that the set of terminals of the minor must be a subset of the
terminals of the original graph. More precisely, $(H,\mathcal{S}')$ is a
signed vertex-minor of $(G,\mathcal{S})$ if $H$ is a vertex-minor of $G$ and
$\mathcal{S}' \subseteq \mathcal{S}$.

Recall that $\mathcal{A}_3$ is the class of graphs built from a three-colored
odd cycle, by adding a terminal for each color, and chords with extremities of
the same color. We define similarly the class $\signed$ of signed vertex-minor
$(G,\mathcal{S})$, where $G \in \mathcal{A}_3$, and $\mathcal{S}$ is the set
of the three terminals in the construction of $G$.

This setting does not affect Lemma~\ref{lemma:main}, and then the following
theorem, close to Theorem~\ref{th:bad-graphs}, can be deduced by the same
proof. Indeed, the proofs in Section~\ref{sec:minor-charac} never create new
terminals when considering vertex-minors, and hence are still valid for signed
vertex-minors.

\begin{theorem}\label{th:signed}
Let $G$ be a graph and $\mathcal{S}$ a set of terminal in $G$. The
system~\eqref{eqn:blocking} is TDI if and only if $(G,\mathcal{S})$ does not have a
signed vertex-minor in $\signed$.\qed
\end{theorem}

\begin{corollary}%not as interesting as the previous dictionary, as we already have the corresponding polynomial algorithm just by testing bipartiteness of the auxiliary graph
The system~\eqref{eqn:blocking} is TDI for $(G,\mathcal{S})$ if and only if it is TDI for every $(G,\mathcal{S}')$, with $\mathcal{S'} \subseteq \mathcal{S}$, $|\mathcal{S}| = 3$.\qed
\end{corollary}

Moreover, all the graphs of $\signed$ are minimal graphs by signed vertex-minors for which system~\eqref{eqn:blocking} is not TDI. Indeed, a potential minor would have the same set of terminals. Moreover, if we contract a vertex, then its two consecutive vertices in the odd cycle become adjacent to two terminals, hence must be deleted. Hence, the minor must be obtained without vertex contraction, and the minimality follows easily. 

\section{Conclusion}

We studied the pairs $(G,\mathcal{S})$ of (graphs, subsets of terminals) for
which the cost of an $\mathcal{S}$-vertex-cut is equal to the maximum packing
of $\mathcal{S}$-paths. We proved that this property for a given $\mathcal{S}$
is polynomially checkable as it reduces to the bipartiteness of an auxiliary
graph. Moreover if this property is true, the minimal $\mathcal{S}$-cut and
maximum path-packing problems can be solved by finding a maximum vertex-capacitated
flow in a smaller graph.

We proved that if $(G,\mathcal{S})$ does not satisfy this property, then there
exists $\mathcal{S}' \subseteq \mathcal{S}$ with $|\mathcal{S}'|=3$ such that
$(G,\mathcal{S}')$ does not satisfy it either. Moreover, each signed
vertex-minor in $\signed$ is a minimal signed vertex-minor obstruction.

Concerning the graphs satisfying the min-max formula for any $\mathcal{S}$, we
proved that they can be recognized in polynomial time, that the list of
vertex-minor obstructions is infinite, but we were unable to provide an
explicit description of this list. We believe that this list is hard to
obtain, and somehow ugly. We also proved that this class of graphs is
interesting as it contains the asteroidal-triple-free graphs.

% One more awesome paper in our personal bibliographies :)

%%bibliography

\newpage

\begin{appendix}
\section*{Appendix to section~\ref{sec:minor-charac}}

The \emph{distance in $C$} between two vertices $u$ and $v$ is the minimum number of arcs in one of the two $(u,v)$-paths in $C$. We denote $d_C(u,v)$ this minimum. We say that $u$ and $v$ are \emph{consecutive} if $d_C(u,v) = 1$. We denote $\repr(u)$ the representant of a vertex $u$. We say that a vertex of $C$ is \emph{bicolored} if its two neighbors in $C$ have distinct colors. Two colors are \emph{adjacent} if there is an edge in $C$ whose ends have these two colors.

Note that the net is a forbidden minor of Mader-Mengerian graphs, and is minimal. We try to find other forbidden minors that do not have a net as vertex-minor. For a graph $H$, we say that $G$ is $H$-free if $H$ is not a vertex minor of $G$.

\begin{lemma}\label{lemma:local-bicolor}
Let $u$ be a bicolored vertex. Let $v$ be a vertex of the same color as $u$. Then, either $G$ contains a net, or every vertex consecutive to $v$ has the color of a vertex consecutive to $u$.
\end{lemma}

\begin{figure}
\begin{center}
\begin{tabular}{cc}
\begin{tikzpicture}[x=0.7cm,y=0.7cm]
\draw[transparent, use as bounding box] (-3,-3) rectangle (3,3);
\draw (75:1.5) arc (75:140:1.5);
\draw[dashed] (140:1.5) arc (140:220:1.5);
\draw (220:1.5) arc (220:320:1.5);
\draw[dashed] (320:1.5) arc (320:360:1.5);
\draw[dashed] (0:1.5) arc (0:75:1.5);
\Vertex{blue} (v) at (90:1.5) {}; \draw (v) node[anchor=south] {$v$};
\Vertex{blue} (u) at (270:1.5) {}; \draw (u) node[anchor=north] {$u$};
\Term{blue} (ru) at (180:0.4) {}; \draw (ru) node[anchor=west] {$\repr(u)$};
\Vertex{red} (u1) at (235:1.5) {}; \draw (u1) node[anchor=south] {$u_1$};
\Term{red} (w1) at (235:2.5) {}; \draw (w1) node[anchor=east] {$\repr(w_1)$};
\Vertex{dgreen} (u2) at (305:1.5) {}; \draw (u2) node[anchor=south] {$u_2$};
\Term{dgreen} (w2) at (305:2.5) {}; \draw (w2) node[anchor=west] {$\repr(w_2)$};
\Vertex{black} (v2) at (125:1.5) {}; \draw (v2) node[anchor=north] {$v'$};
\Term{black} (rv2) at (125:2.5) {}; \draw (rv2) node[anchor=south] {$\repr(v')$};
\draw (u1) -- (w1) (u2) -- (w2) (u) -- (ru) -- (v) (v2) -- (rv2);
\end{tikzpicture}
&
\begin{tikzpicture}[x=0.7cm,y=0.7cm]
\draw[transparent, use as bounding box] (-3,-3) rectangle (3,3);
\draw[very nearly transparent] (0,0) circle (1.5);
\Vertex{blue} (v) at (90:1.5) {}; \draw (v) node[anchor=south] {$v$};
\Vertex{red} (u1) at (235:1.5) {}; \draw (u1) node[anchor=south west] {$u_1$};
\Term{red} (w1) at (235:2.5) {}; \draw (w1) node[anchor=east] {$\repr(w_1)$};
\Vertex{dgreen} (u2) at (305:1.5) {}; \draw (u2) node[anchor=south east] {$u_2$};
\Term{dgreen} (w2) at (305:2.5) {}; \draw (w2) node[anchor=west] {$\repr(w_2)$};
\Term{black} (rv2) at (125:2.5) {}; \draw (rv2) node[anchor=south] {$\repr(v')$};
\draw (w1) -- (u1) -- (u2) -- (w2) (u1) -- (v) -- (u2) (v) -- (rv2);
\end{tikzpicture}
\\
$(a)$ & $(b)$ \\
\begin{tikzpicture}[x=0.8cm,y=0.8cm]
\draw[transparent, use as bounding box] (-3,-3) rectangle (3,3);
\draw (100:1.5) arc (100:300:1.5);
\draw[dashed] (300:1.5) arc (300:360:1.5);
\draw[dashed] (0:1.5) arc (0:100:1.5);
\Vertex{blue} (v) at (120:1.5) {}; \draw (v) node[anchor=south east] {$v$};
\Vertex{blue} (u) at (240:1.5) {}; \draw (u) node[anchor=north east] {$u$};
\Term{blue} (ru) at (180:0.4) {}; \draw (ru) node[anchor=west] {$\repr(u)$};
\Vertex{red} (u1) at (200:1.5) {}; \draw (u1) node[anchor=west] {$u_1$};
\Term{red} (w1) at (200:2.5) {}; \draw (w1) node[anchor=east] {$\repr(w_1)$};
\Vertex{dgreen} (u2) at (280:1.5) {}; \draw (u2) node[anchor=south] {$u_2$};
\Term{dgreen} (w2) at (280:2.5) {}; \draw (w2) node[anchor=north] {$\repr(w_2)$};
\Vertex{black} (v2) at (160:1.5) {}; \draw (v2) node[anchor=north east] {$v'$};
\Term{black} (rv2) at (160:2.5) {};\draw (rv2) node[anchor=south east] {$\repr(v')$};
\draw (u1) -- (w1) (u2) -- (w2) (u) -- (ru) -- (v) (v2) -- (rv2);
\end{tikzpicture}
&
\begin{minipage}{0.2\textwidth}
\end{minipage}
\begin{tikzpicture}[x=0.8cm,y=0.8cm]
\draw[transparent, use as bounding box] (-3,-3) rectangle (3,3);
\draw[very nearly transparent] (0,0) circle (1.5);
\draw (160:1.5) arc (160:200:1.5);
\Vertex{red} (u1) at (200:1.5) {}; \draw (u1) node[anchor=west] {$u_1$};
\Term{red} (w1) at (200:2.5) {}; \draw (w1) node[anchor=east] {$\repr(w_1)$};
\Vertex{dgreen} (u2) at (280:1.5) {}; \draw (u2) node[anchor=south] {$u_2$};
\Term{dgreen} (w2) at (280:2.5) {}; \draw (w2) node[anchor=north] {$\repr(w_2)$};
\Vertex{black} (v2) at (160:1.5) {}; \draw (v2) node[anchor=north east] {$v'$};
\Term{black} (rv2) at (160:2.5) {};\draw (rv2) node[anchor=south east] {$\repr(v')$};
\draw (w1) -- (u1) -- (u2) -- (w2) (rv2) -- (v2) -- (u2);
\end{tikzpicture}
\\
$(c)$ & $(d)$
\end{tabular}
\end{center}
\caption{Illustration for Lemma~\ref{lemma:local-bicolor}.}
\label{fig:lemma1}
\end{figure}
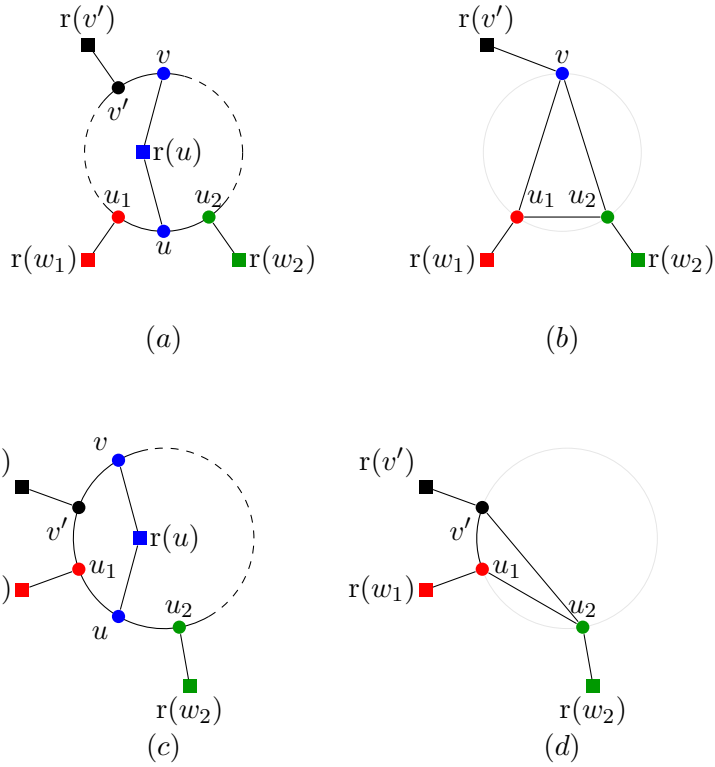

\begin{figure}[htbp]
\begin{center}
\begin{tabular}{cc}
\begin{minipage}{0.45\textwidth}
\begin{center}
\begin{tikzpicture}[x=0.7cm,y=0.7cm]
\draw[transparent,use as bounding box] (-3,-3) rectangle (3,3);
\draw[dashed] (0:2) arc (0:60:2);
\draw (60:2) arc (60:120:2);
\draw[dashed] (120:2) arc (120:180:2);
\draw (180:2) arc (180:240:2);
\draw[dashed] (240:2) arc (240:300:2);
\draw (300:2) arc (300:360:2);
\Vertex{blue} (o) at (0,0) {}; \draw (o) node[anchor=south east] {$\repr(u_1)$};
\Vertex{blue} (u1) at (90:2) {}; \draw (u1) node[anchor=south] {$u_1$};
\Vertex{blue} (u2) at (210:2) {}; \draw (u2) node[anchor=west] {$u_2$};
\Vertex{blue} (u3) at (330:2) {}; \draw (u3) node[anchor=east] {$u_3$};
\Term{red} (rv1) at (90:2.9) {};     \draw (rv1) node[anchor=south] {$\repr(v_1)$};
\Term{dgreen} (rv2) at (210:2.9) {}; \draw (rv2) node[anchor=north east] {$\repr(v_2)$};
\Term{black} (rv3) at (330:2.9) {};  \draw (rv3) node[anchor=north west] {$\repr(v_3)$};
\Vertex{red} (v1) at (110:2) {}; \draw (v1) node[anchor=east] {$v_1$};
\Vertex{red} (w1) at (70:2) {};
\Vertex{dgreen} (v2) at (230:2) {}; \draw (v2) node[anchor=north] {$v_2$};
\Vertex{dgreen} (w2) at (190:2) {};
\Vertex{black} (v3) at (350:2) {}; \draw (v3) node[anchor=west] {$v_3$};
\Vertex{black} (w3) at (310:2) {};
\draw (u1) -- (o) -- (u2);
\draw (o) -- (u3);
\draw (v1) -- (rv1) -- (w1);
\draw (v2) -- (rv2) -- (w2);
\draw (v3) -- (rv3) -- (w3);
\end{tikzpicture}
\end{center}
\end{minipage}
&
\begin{minipage}{0.45\textwidth}
\begin{center}
\begin{tikzpicture}[x=0.7cm,y=0.7cm]
\draw[transparent,use as bounding box] (-3,-3) rectangle (3,3);
\draw[very nearly transparent] (0,0) circle (2);
\Vertex{red} (u1) at (110:2) {}; \draw (u1) node[anchor=east] {$v_1$};
\Vertex{dgreen} (u2) at (230:2) {}; \draw (u2) node[anchor=north west] {$v_2$};
\Vertex{black} (u3) at (350:2) {}; \draw (u3) node[anchor=north east] {$v_3$};
\Term{red} (rv1) at (90:2.9) {};     \draw (rv1) node[anchor=south] {$\repr(v_1)$};
\Term{dgreen} (rv2) at (210:2.9) {}; \draw (rv2) node[anchor=north east] {$\repr(v_2)$};
\Term{black} (rv3) at (330:2.9) {};  \draw (rv3) node[anchor=north west] {$\repr(v_3)$};
\draw (rv1) -- (u1) -- (u2) -- (u3) -- (rv3) (u1) -- (u3) (rv2) -- (u2);
\end{tikzpicture}
\end{center}
\end{minipage}
\\
$(a)$ & $(b)$
\end{tabular}
\end{center}
\caption{Illustration for Lemma~\ref{lemma:bicolor}.}
\label{fig:bicolor}
\end{figure}
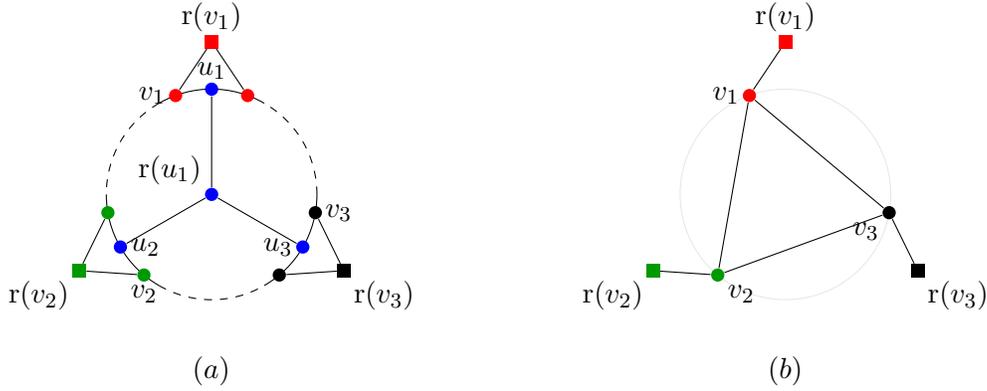

\begin{proof}
Let $u_1$ and $u_2$ be adjacent to $u$ in $C$, $v'$ is adjacent to $v$, and $u_1$, $u_2$ and $v'$ have distinct colors. First, suppose that $d_C(u,v) \geq 3$. There are two cases.

If $d_C(v',u) \geq 3$ (Figure~\ref{fig:lemma1}, $a$), let $G'$ be the graph obtained by contracting $u$ and $\repr(u)$ and by deleting all the vertices except $\repr(u_1)$, $\repr(u_2)$, $u_1$, $u_2$ and $v$. $G'$ is a net (Figure~\ref{fig:lemma1}, $b$).

If $d_C(v',u) = 2$ (Figure~\ref{fig:lemma1}, $c$), we may assume $v'u_1 \in E(C)$. Let $G'$ be the graph obtained from $G$ by contracting $u$, $v$ and $\repr(u)$ and deleting every other vertex except $v'$, $\repr(v')$, $\repr(u_1)$, $\repr(u_2)$, $u_1$ and $u_2$. Then $G'$ is a net (Figure~\ref{fig:lemma1}, $d$).

Now suppose that $d_C(u,v) = 2$. We may assume that $u_1$ is adjacent to $v$. Then the graph obtained from $G$ by contracting $\repr(u)$ and deleting every vertex except $u$, $v$, $u_1$, $u_2$, $v'$ and $\repr(u_1)$, is a net.
\end{proof}

\begin{lemma}\label{lemma:bicolor}
Every color is adjacent to at most two other colors, or $G$ contains a net vertex-minor.
\end{lemma}

\begin{proof}
Let $R$ be any color. By applying iteratively Lemma~\ref{lemma:local-bicolor}, if there is a vertex of color $R$ whose two consecutive vertices have distinct colors, then either $G$ contains a net, or $R$ is adjacent to exactly two colors.

Otherwise, each vertex in $C$ is consecutive to two vertices of the same color. Suppose that there are three vertices $u_1$, $u_2$, $u_3$ in $C$ of color $R$, such that their neighbors have three different colors. let $v_1$, $v_2$ and $v_3$ be the vertices following $u_1$, $u_2$, $u_3$ respectively in $C$ (Figure~\ref{fig:bicolor}, $a$). Then, by contracting $u_1$, $u_2$, $u_3$, $\repr(u_1)$ and deleting all the vertices except $v_1$, $v_2$, $v_3$ and their representants, we obtain a net (Figure~\ref{fig:bicolor}, $b$). 
\end{proof}

From now, we suppose that $G$ does not have a net minor. We define the \emph{graph of colors}, whose vertices are the colors, by the adjacency relation introduced above. By Lemma~\ref{lemma:bicolor}, the graph of color has maximum degree two. By connexity, it is either a cycle or a path. We index the colors from $1$ to $n$, following the order defined by the path or the cycle. Thus, each edge of $C$ has extremities of colors $i$ and $i+1$, or $1$ and $n$. We have the following immediate consequence.

\begin{lemma}\label{lemma:odd-colors}
Let $G$ be net-free. The number $n$ of colors is odd, and the graph of colors is a cycle.
\end{lemma}

\begin{proof}
Suppose not. Then $C$ has a proper $2$-coloring (following the parity of the colors), thus is even, contradicting the assumption. 
\end{proof}

\begin{lemma}\label{lemma:bicolored-vertices}
Let $G$ be net-free. Every color contains a bicolored vertex.
\end{lemma}

\begin{figure}
\begin{center}
\begin{tikzpicture}[x=0.8cm,y=0.8cm]
\fill[black,nearly transparent] (0.5,2) arc (0:180:0.5 and 1);
\fill[black,semitransparent] (-0.5,2) arc (180:360:0.5 and 1);
\fill[black,nearly transparent,rotate=72] (0,2) ellipse (0.5 and 1);
\fill[black,semitransparent,rotate=144] (0,2) ellipse (0.5 and 1);
\fill[black,nearly transparent,rotate=216] (0,2) ellipse (0.5 and 1);
\fill[black,semitransparent,rotate=288] (0,2) ellipse (0.5 and 1);
\Vertex{black} (a1) at (90:1.55) {}; 
\Vertex{black} (a2) at (90:2.15) {};
\Vertex{black} (a3) at (90:2.55) {};
\Vertex{black} (b1) at (18:1.5) {};
\Vertex{black} (b2) at (18:2) {};
\Vertex{black} (b3) at (18:2.5) {};
\Vertex{black} (c1) at (306:1.5) {};
\Vertex{black} (c2) at (306:2) {};
\Vertex{black} (c3) at (306:2.5) {};
\Vertex{black} (d1) at (234:1.5) {};
\Vertex{black} (d2) at (234:2) {};
\Vertex{black} (d3) at (234:2.5) {};
\Vertex{black} (e1) at (162:1.5) {};
\Vertex{black} (e2) at (162:2) {};
\Vertex{black} (e3) at (162:2.5) {};
\draw (a3) -- (b3) -- (a2) -- (b2) -- (c3) -- (b1) -- (c2) -- (d2) -- (c1) -- (d1) 
   -- (e1) -- (d3) -- (e2) -- (a1) -- (e3) -- (a3);
\draw (90:3.5) node {$1$};
\draw (162:3.5) node {$2$};
\draw (234:3.5) node {$3$};
\draw (306:3.5) node {$4$};
\draw (18:3.5) node {$5$};
\end{tikzpicture}
\end{center}
\caption{Illustration for Lemma~\ref{lemma:bicolored-vertices}, each ellipse represents a color.}
\label{fig:bicolored-vertices}
\end{figure}
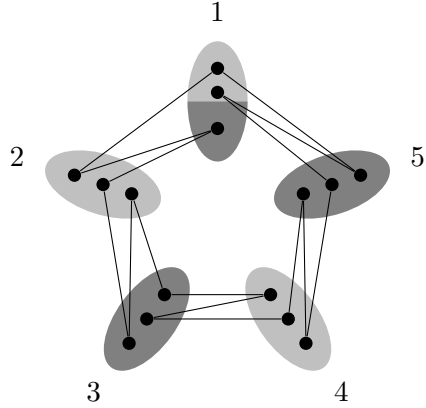

\begin{proof}
Without loss of generality, it is sufficient to prove that there is a bicolored vertex of color $1$. Let $U$ be the vertices of color $1$ adjacent to a vertex of color $n$ and $U'$ be the other vertices of color $1$. Let $W$ the vertices of $C$ having an odd color minus $U'$, and $B$ its complement in $V(C)$(see Figure~\ref{fig:bicolored-vertices}). $B$ does not contain two consecutive vertices of $C$, and $(W,B)$ cannot be a proper two coloring of $C$. Thus there is an edge in $C$ with both extremities in $W$. But this can only be an edge between $U$ and color $n$, hence there is a bicolored vertex in $U$. 
\end{proof}

\begin{lemma}\label{lemma:7colors}
If $G$ is net-free, the number $n$ of colors is at most $5$.
\end{lemma}

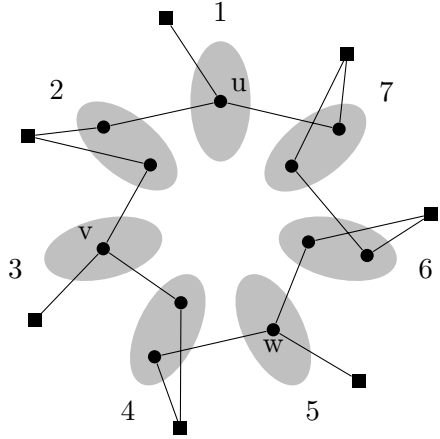
\begin{figure}
\begin{center}
\begin{tikzpicture}[x=0.8cm,y=0.8cm]
\fill[nearly transparent] (0,2) ellipse (0.5 and 1);
\fill[nearly transparent,rotate=51] (0,2) ellipse (0.5 and 1);
\fill[nearly transparent,rotate=103] (0,2) ellipse (0.5 and 1);
\fill[nearly transparent,rotate=154] (0,2) ellipse (0.5 and 1);
\fill[nearly transparent,rotate=206] (0,2) ellipse (0.5 and 1);
\fill[nearly transparent,rotate=257] (0,2) ellipse (0.5 and 1);
\fill[nearly transparent,rotate=308] (0,2) ellipse (0.5 and 1);
\Vertex{black} (a) at (90:2) {};\Vertex{black} (c) at (193:2) {};
\Vertex{black} (e) at (296:2) {};
\Vertex{black} (b1) at (141:1.5) {};\Vertex{black} (b2) at (141:2.5) {};
\Vertex{black} (d1) at (244:1.5) {};\Vertex{black} (d2) at (244:2.5) {};
\Vertex{black} (f1) at (347:1.5) {};\Vertex{black} (f2) at (347:2.5) {};
\Vertex{black} (g1) at (38:1.5) {};\Vertex{black} (g2) at (38:2.5) {};
\draw (90:3.5) node {$1$};\draw (141:3.5) node {$2$};\draw (193:3.5) node {$3$};
\draw (244:3.5) node {$4$};\draw (296:3.5) node {$5$};\draw (347:3.5) node {$6$};
\draw (38:3.5) node {$7$};
\draw (a) node[anchor=south west] {u};
\draw (c) node[anchor=south east] {v};
\draw (e) node[anchor=north] {w};
\Term{black} (r1) at (105:3.5) {};\Term{black} (r2) at (156:3.5) {};
\Term{black} (r3) at (208:3.5) {};\Term{black} (r4) at (259:3.5) {};
\Term{black} (r5) at (311:3.5) {};\Term{black} (r6) at (2:3.5) {};
\Term{black} (r7) at (53:3.5) {};
\draw (r1) -- (a) -- (b2) -- (r2) -- (b1) -- (c) -- (r3)
      (c) -- (d1) -- (r4) -- (d2) -- (e) -- (r5)
      (e) -- (f1) -- (r6) -- (f2) -- (g1) -- (r7) -- (g2) -- (a);
\end{tikzpicture}
\end{center}
\caption{Illustration for Lemma~\ref{lemma:7colors}. The representant of each color is drawn as a square.}
\label{fig:7colors}
\end{figure}

\begin{proof}
By contradiction. Let $u$, $v$ and $w$ be bicolored vertices of colors $1$, $3$ and $5$ respectively (see Figure~\ref{fig:7colors}). Contract every vertex of other colors, and delete every remaining vertex except $u$, $v$, $w$ and their representants. If $n > 5$, the $6$-vertices graph obtained by this way is a net.
\end{proof}

\begin{lemma}\label{lemma:bicol-consec}
If $G$ is net-free and $n=5$, there are no two consecutive bicolored vertices in $C$.
\end{lemma}

\begin{proof}
By contradiction. Let $u$ be a bicolored vertex, of color $1$, and $v$ its bicolored neighbor of color $2$. Let $w$ be a bicolored vertex of color $4$. Then, the graph obtained by contracting vertices of colors $3$ and $5$, and deleting all the vertices of colors $1$, $2$ and $4$, except $u$, $v$, $w$ and their representants, is a net. 
\end{proof}

\begin{lemma}\label{lemma:parity}
Suppose $G$ is net-free. The number of edges between any two color classes is zero or odd.
\end{lemma}

\begin{proof}
Choose two adjacent colors, and remove from $C$ every edge between these two colors. Every path thus obtained has its both extremities in the same color class, or in the two chosen colors. So every path has an even length. As $C$ is odd, this proves that we removed an odd number of edges. 
\end{proof}

\begin{lemma}\label{lemma:odd-sequence}
Suppose $G$ is net-free. If $n=5$, there is a maximal sequence of consecutive edges in $C$ between any two given adjacent colors of length $2k+1$, for some $k \geq 1$.
\end{lemma}

\begin{proof}
Consider the subpaths of $C$ obtained by keeping only the edges between colors $1$ and $2$. Either there is a net minor or each of these paths has length at least $2$ by Lemma~\ref{lemma:bicol-consec}. Then, by Lemma~\ref{lemma:parity}, there is a path of odd length, proving the lemma.
\end{proof}

\begin{lemma}\label{lemma:5colors}
If $n=5$, then $G$ is not a minimally excluded graph by vertex minor.
\end{lemma}

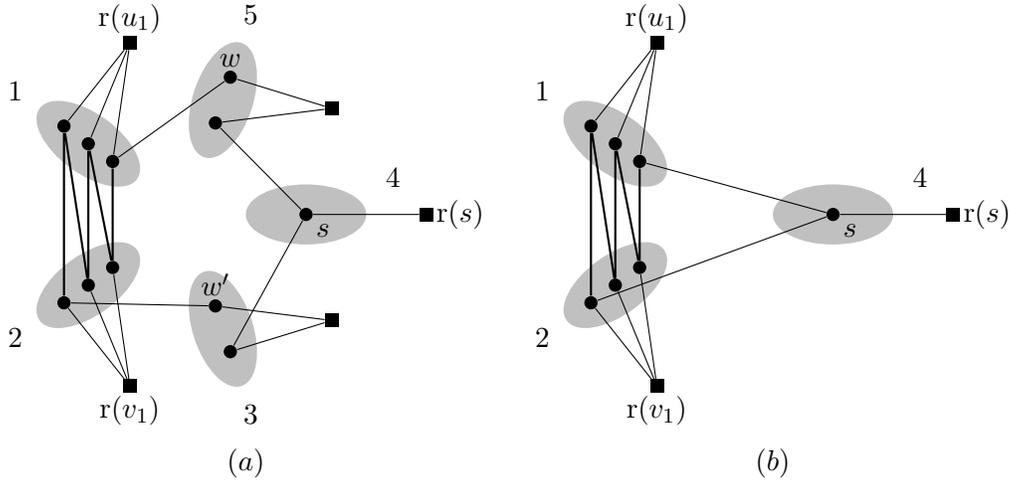
\begin{figure}
\begin{center}
\begin{tabular}{cc}
\begin{tikzpicture}[x=0.8cm,y=0.8cm]
\fill[nearly transparent] (2,0) ellipse (1 and 0.5);
\fill[nearly transparent,rotate=72] (2,0) ellipse (1 and 0.5);
\fill[nearly transparent,rotate=144] (2,0) ellipse (1 and 0.5);
\fill[nearly transparent,rotate=216] (2,0) ellipse (1 and 0.5);
\fill[nearly transparent,rotate=288] (2,0) ellipse (1 and 0.5);
\Vertex{black} (a1) at (144:1.5) {};\Vertex{black} (a2) at (144:2) {};
\Vertex{black} (a3) at (144:2.5) {};\Vertex{black} (b1) at (216:1.5) {};
\Vertex{black} (b2) at (216:2) {};\Vertex{black} (b3) at (216:2.5) {};
\Vertex{black} (c1) at (288:1.6) {};\Vertex{black} (c2) at (288:2.4) {};
\Vertex{black} (d1) at (0:2) {};
\Vertex{black} (e1) at (72:1.6) {};\Vertex{black} (e2) at (72:2.4) {};
\Term{black} (a) at (108:3) {};\Term{black} (b) at (252:3) {};
\Term{black} (c) at (324:3) {};\Term{black} (d) at (0:4) {};
\Term{black} (e) at (36:3) {};
\draw (e2) node[anchor=south] {$w$}; \draw (c1) node[anchor=south] {$w'$};
\draw (d1) node[anchor=north west] {$s$}; \draw (d) node[anchor=west] {$\repr(s)$};
\draw (a) node[anchor=south] {$\repr(u_1)$}; \draw (b) node[anchor=north] {$\repr(v_1)$};
\draw (d) -- (d1) -- (e1) -- (e) -- (e2) -- (a1) -- (a) -- (a2);
\draw[thick] (a1) -- (b1) -- (a2) -- (b2) -- (a3) -- (b3);
\draw (b2) -- (b) -- (b3) -- (c1) -- (c) -- (c2) -- (d1);
\draw (a) -- (a3) (b) -- (b1);
\draw (10:3.5) node {$4$};
\draw (72:3.5) node {$5$}; 
\draw (144:3.5) node {$1$}; 
\draw (216:3.5) node {$2$}; 
\draw (288:3.5) node {$3$}; 
\end{tikzpicture}
&
\begin{tikzpicture}[x=0.8cm,y=0.8cm]
\fill[nearly transparent] (2,0) ellipse (1 and 0.5);
\fill[nearly transparent,rotate=144] (2,0) ellipse (1 and 0.5);
\fill[nearly transparent,rotate=216] (2,0) ellipse (1 and 0.5);
\Vertex{black} (a1) at (144:1.5) {};\Vertex{black} (a2) at (144:2) {};
\Vertex{black} (a3) at (144:2.5) {};\Vertex{black} (b1) at (216:1.5) {};
\Vertex{black} (b2) at (216:2) {};\Vertex{black} (b3) at (216:2.5) {};
\Vertex{black} (d1) at (0:2) {};
\Term{black} (a) at (108:3) {};\Term{black} (b) at (252:3) {};
\Term{black} (d) at (0:4) {};
\draw (d1) node[anchor=north west] {$s$}; \draw (d) node[anchor=west] {$\repr(s)$};
\draw (a) node[anchor=south] {$\repr(u_1)$}; \draw (b) node[anchor=north] {$\repr(v_1)$};
\draw (d) -- (d1) -- (a1) -- (a) -- (a2);
\draw[thick] (a1) -- (b1) -- (a2) -- (b2) -- (a3) -- (b3);
\draw (b2) -- (b) -- (b3) -- (d1);
\draw (a) -- (a3) (b) -- (b1);
\draw (10:3.5) node {$4$};
\draw (144:3.5) node {$1$}; 
\draw (216:3.5) node {$2$}; 
\end{tikzpicture}\\
$(a)$ & $(b)$
\end{tabular}
\end{center}
\caption{Illustration for Lemma~\ref{lemma:5colors}}
\label{fig:5colors}
\end{figure}

\begin{proof}
If $G$ contains a net minor, it is clearly not minimal. Suppose it does not.
  
Let $u_1$, $v_1$, $u_2$, \ldots $u_k$, $v_k$ be a maximum subpath of $C$ of odd length between colors $1$ and $2$. By Lemma~\ref{lemma:odd-sequence}, we have $k \geq 2$. Let $w$ be the vertex of color $5$ adjacent to $u_1$, and $w'$ the vertex of color $3$ adjacent to $v_k$. Let $s$ be a bicolored vertex of color $4$ (see Figure~\ref{fig:5colors}, $a$).

Consider the graph obtained by contracting vertices of colors $3$ and $5$, and deleting all the other vertices except $u_1$, $v_1$, $u_2$, \ldots $u_k$, $v_k$, $s$, $\repr(u_1)$, $\repr(u_2)$ and $\repr(s)$. This graph (Figure~\ref{fig:5colors}, $b$) is composed of a cycle of length $2k+1$ plus three terminals $\repr(u_1)$, $\repr(u_2)$ and $\repr(s)$. It obviously checks the condition for being an excluded graph. Thus $G$ is not minimal.
\end{proof}

\end{appendix}

\end{document}